\documentclass[11pt]{article}
\usepackage{amsmath, amssymb, amsthm}
\usepackage{graphicx}
\usepackage{sham}
\usepackage{mathrsfs}

\newcommand{\ZX}{\Int_0^+[x]}

\date{}

\renewcommand{\baselinestretch}{1.25}

\title{\Large\bf Counting number of factorizations of a natural number}

\author{Shamik Ghosh}

\begin{document}

\maketitle

\address{Department of Mathematics, Jadavpur University,
Kolkata - 700 032, India.}

\vspace{0.5em} \email{sghosh@math.jdvu.ac.in}

\subjclassetc{Primary 11N25}{Unordered factorization, partition
function.}

%%%%%%%%%%%%%%%%%%%%%%%%%%%%%%%%%%%%%%%%%%%%%%%%%%%%%%%%%%%%%%%%%%

\renewcommand{\baselinestretch}{1}
\begin{abstract}
{\small In this note we describe a new method of counting the
number of unordered factorizations of a natural number by means
of a generating function and a recurrence relation arising from
it, which improves an earlier result in this direction.}
\end{abstract}

%%% ----------------------------------------------------------------------

\renewcommand{\baselinestretch}{1.5}

%%% ----------------------------------------------------------------------

%%% ----------------------------------------------------------------------

\section{Introduction}

Consider the natural number $18$. It has $4$ distinct
``factorizations'', namely
$$18=2.3.3=2.9=3.6=18.$$
Similarly, there are $9$ ways of factorizing $36$: $36=2.2.3.3
=2.2.9=2.3.6=3.3.4=2.18=3.12=4.9=6.6=36$. Our problem is to find
this number for any natural number $n$. Since we are not
distinguishing between $2\cdot 9$ and $9\cdot 2$, such a
factorization is called {\em unordered}. A
{\em\textbf{partition}} of a natural number $n$ is a
representation of $n$ as the sum of any number of positive
integral parts, where the order of the parts is irrelevant. The
number of such partitions of $n$ is known as the
{\em\textbf{partition function}} and is denoted by $p(n)$.
Likewise the function $p^*(n)$ denotes the number of ways of
expressing $n$ as a product of positive integers greater than
$1$, the order of the factors in the product being irrelevant.
For convenience, $p^*(1)$ is assumed to be $1$. Clearly $p^*(n)$
is the number of unordered factorizations of $n$. In 1983, Hughes
and Shalit \cite{HuS} obtained a bound for $p^*(n)$, namely,
$p^*(n)\leqslant 2n^{\surd 2}$ which was then improved to
$p^*(n)\leqslant n$ by Mattics and Dodd \cite{MD} in 1986. By
this time Canfield, Erd\"{o}s and Pomerance \cite{CEP} modified a
result of Oppenheim regarding the maximal order of $p^*(n)$. They
obtained another bound for $p^*(n)$ and described an algorithm
for it. An average estimate for $p^*(n)$ was given by Oppenheim
\cite{OP} which was also proved independently by Szekeres and
Turan \cite{ST}. Finally, in 1991, Harris and Subbarao \cite{HS}
came with a generating function and a recursion formula for
$p^*(n)$. One may consider \cite{D} and \cite{GN} for some
problems associated with $p^*(n)$. For a list of values and
computer programming one may consider the website:
http://www.research.att.com/cgi-bin/access.cgi/as/njas/sequences/
 (sequence no. A001055).

\vspace{1em} In this note, we describe a new method for counting
$p^*(n)$ and obtain a generating function for it which is
followed by a recurrence relation that generalizes the one given
by Harris and Subbarao \cite{HS} as well as the one given by
Euler for $p(n)$. The final recursion formula improves the one
given in \cite{HS} as it contains less number of terms. It is
important to note that we wish to develop an algebraic approach
to the problem which might be helpful for other similar
situations in future. Also we note some errors in describing an
equivalent form of the recurrence relation in \cite{HS}.
Throughout the note we denote set of all natural numbers,
non-negative integers, integers and rational numbers by $\Nat,\
\Int_0^+,\ \Int,\ \Rat$ respectively.

\section{Representation of numbers by polynomials}

Consider the monoid $(\Nat,\cdot)$ of natural numbers under usual
multiplication. For any natural number $n$, let $S(n)$ be the
submonoid of $(\Nat,\cdot)$, generated by the set
 of prime factors of $n$, i.e., if the prime factorization of $n$ is
\begin{equation}\label{eq:n}
n=p_1^{n_1}p_2^{n_2}\ldots p_k^{n_k},
\end{equation}
where $p_i$ are distinct primes, $n_i\in\Nat$ for all
$i=1,2,\ldots ,k,\ (k\in\Nat)$, then
$$S(n)=\Set{p_1^{r_1}p_2^{r_2}\ldots
p_k^{r_k}\in\Nat}{r_i\in\Int_0^+,\ i=1,2,\ldots ,k}.$$ We show
that $S(n)$ has an interesting algebraic structure. Define the
partial ordering $\leq^\cdot$ on $S(n)$ by
$$a\leq^\cdot b\ \Longleftrightarrow\ a\ \textrm{
divides }\ b.$$ This ordering on $S(n)$ is, in fact, a lattice
ordering where $a\vee b=\lcm (a,b)$ and $a\wedge b=\gcd (a,b)$
for all $a,b\in S(n)$. Moreover this lattice is distributive and
bounded below by $1$. A monoid $S$ is called a
{\em\textbf{lattice-ordered semigroup}} if it has a lattice
ordering that satisfies $a(b\vee c)=ab\vee ac$ and $(b\vee
c)a=ba\vee ca$, for all $a,b,c\in S$. Now for all $a,b,c\in
S(n)$, $a\{\lcm (b,c)\}=\lcm (ab,ac)$. So we have the following
proposition:

\begin{prop}
For any natural number $n$, $(S(n),\cdot,\leq^\cdot)$ is a
lattice-ordered semigroup.
\end{prop}

\begin{defn}\label{d:fxn}
Now corresponding to each natural number $n$ we associate a
polynomial in the polynomial semiring $\Int_0^+[x]$ as
$$f(x;n)=n_1+n_2x+n_3x^2+\cdots +n_kx^{k-1},$$
where $n$ has the prime factorization (\ref{eq:n}).

Next we define a binary relation $\leqq$ on $\Int_0^+[x]$ as
follows:

\noindent Let $f(x)=a_0+a_1x+a_2x^2+\cdots +a_mx^m$ and
$g(x)=b_0+b_1x+b_2x^2+\cdots +b_nx^n$. Then
$$f(x)\leqq g(x)\ \Longleftrightarrow\ m\leqslant n\ \textrm{ and }\
a_i\leqslant b_i\ \textrm{ for all }\ i=0,1,2,\ldots ,m.$$
Clearly $\leqq$ is a partial ordering on $\Int_0^+[x]$. Finally,
let us denote the set of all polynomials in $\Int_0^+[x]$ of
degree less than $k$ by $P_k[x]$.
\end{defn}

\begin{thm}\label{t:losem}
$(P_k[x],+,\leqq)$ is a lattice-ordered semigroup which is
isomorphic to $(S(n),\cdot,\leq^\cdot)$, where $n$ has the prime
factorization (\ref{eq:n}).
\end{thm}

\begin{proof}
Let $f(x)=a_1+a_2x+a_3x^2+\cdots +a_kx^{k-1},\
g(x)=b_1+b_2x+b_3x^2+\cdots +b_kx^{k-1}\in P_k[x]$. Then it is
routine to verify that $f\vee g=c_1+c_2x+c_3x^2+\cdots
+c_kx^{k-1}$ and $f\wedge g=d_1+d_2x+d_3x^2+\cdots +d_kx^{k-1}$,
where $c_i=\max\{a_i,b_i\}$ and $d_i=\min\{a_i,b_i\}$. Thus
$(P_k[x],\leqq)$ is a lattice. Obviously, $(P_k[x],+)$ is an
abelian monoid where the identity element is the zero polynomial.
Now choose $f(x)=\sum\limits_{i=1}^k a_ix^{i-1},\
g(x)=\sum\limits_{i=1}^k b_ix^{i-1},\ h(x)=\sum\limits_{i=1}^k
c_ix^{i-1}\in P_k[x]$. Then $f+(g\vee h)=\sum\limits_{i=1}^k
\big(a_i+\max\{b_i,c_i\}\big)x^{i-1}= \sum\limits_{i=1}^k
\big(\max\{a_i+b_i,a_i+c_i\}\big)x^{i-1} = (f+g)\vee (f+h)$.
Therefore $(P_k[x],+,\leqq)$ is a lattice-ordered semigroup.

Now define a map $\Map{\psi}{S(n)}{P_k[x]}$ by $\psi(m)=f(x;m)$
for all $m\in S(n)$. That $\psi$ is bijective follows from
Definition \ref{d:fxn}. Let $m_1=\prod\limits_{i=1}^k
p_i^{r_{1i}},\ m_2=\prod\limits_{i=1}^k p_i^{r_{2i}}\in S(n)$.
Then
$$\begin{array}{ll}
&m_1\leq^\cdot m_2\\
\Longleftrightarrow & m_1\ \textrm{ divides }\ m_2\\
\Longleftrightarrow & r_{1i}\leqslant r_{2i}$ for each
$i=1,2,\ldots,k\\
\Longleftrightarrow & r_{11}+r_{12}x+r_{13}x^2+\cdots
+r_{1k}x^{k-1}\leqq r_{21}+r_{22}x+r_{23}x^2+\cdots
+r_{2k}x^{k-1}\\
\Longleftrightarrow & f(x;m_1)\leqq f(x;m_2)\\
\Longleftrightarrow & \psi(m_1)\leqq \psi(m_2).
\end{array}$$

Also $m_1m_2= \prod\limits_{i=1}^k p_i^{r_{1i}+r_{2i}}$. Then
$f(x;m_1m_2)=\sum\limits_{i=1}^k (r_{1i}+r_{2i})x^{i-1}
=\sum\limits_{i=1}^k {r_{1i}}x^{i-1}+\sum\limits_{i=1}^k
{r_{2i}}x^{i-1}=f(x;m_1)+f(x;m_2)$. Thus
$\psi(m_1m_2)=\psi(m_1)+\psi(m_2)$. Therefore $\psi$ is an
isomorphism, i.e., $(S(n),\cdot,\leq^\cdot)\ \cong\
(P_k[x],+,\leqq)$, as required.
\end{proof}

\begin{defn}
 For
any $f(x)\in\Int_0^+[x]$, let $p(f)$ denote the number of
partitions of the polynomial $f(x)$ in terms of addition of
polynomials (not all distinct) in $\Int_0^+[x]$, where the order
of addition is irrelevant. We assume $p(0)=1$.
\end{defn}

For example, the distinct partitions of the polynomial $2+x$ in
$\Int_0^+[x]$ are
\begin{eqnarray*}
2+x&=&(1)+(1)+(x)\\
&=&(1)+(1+x)\\
&=&(2)+(x)\\
&=&(2+x)
\end{eqnarray*}

\noindent Note that $f(x;12)=2+x$ and compare the above partitions
with usual factorizations: $12=2.2.3=2.(2.3)=2^2.3=12$.

\begin{thm}
For any natural number $n$, $p^*(n)=p(f(x;n))$.
\end{thm}

\begin{proof}
Let $n\in\Nat$ and $F(n)$ denote the set of all factors of $n$.
Then $(F(n),\leq^\cdot)$ is a sublattice of $S(n)$. On the other
hand, define\footnote{$S(f(x))$ is called the
{\em\textbf{section}} of $f(x)$ in $\Int_0^+[x]$.} the set
$$S(f(x))=\Set{g(x)\in\Int_0^+[x]}{g(x)\leqq f(x)}.$$
Then $S(f(x))$ is a sublattice of $P_k[x]$, where $f(x)=f(x;n)$
and $n$ has the prime factorization (\ref{eq:n}). By Theorem
\ref{t:losem}, it follows that the restriction of the map $\psi$
on $F(n)$ is a lattice isomorphism from $(F(n),\leq^\cdot)$ onto
$\big(S(f(x;n)),\leqq\big)$. Indeed, let
$m=p_1^{r_1}p_2^{r_2}\ldots p_k^{r_k}\in F(n)$ for some
$r_1,r_2,\ldots,r_k\in\Int_0^+$. Since $m$ is a factor of $n$, we
have $r_i\leqslant n_i$ for each $i=1,2,\ldots,k$. Hence
$\psi(m)=f(x;m)=r_1+r_2x+r_3x^2+\cdots +r_kx^{k-1} \leqq
n_1+n_2x+n_3x^2+\cdots +n_kx^{k-1} = f(x;n)$ which implies
$\psi(m)\in S(f(x;n))$. Conversely, let $g(x)\in S(f(x;n))$. Then
$\deg g(x)\leqslant \deg f(x)=k-1$. Let
$g(x)=b_1+b_2x+b_3x^2+\cdots +b_kx^{k-1}$ for some
$b_1,b_2,\ldots,b_k\in\Int_0^+$. Then $b_i\leqslant n_i$ for each
$i=1,2,\ldots ,k$ and $g(x)=g(x;p_1^{b_1}p_2^{b_2}\ldots
p_k^{b_k})=\psi(p_1^{b_1}p_2^{b_2}\ldots p_k^{b_k})$, which
implies $\psi(F(n))=S(f(x;n))$. Thus we have $(F(n),\leq^\cdot)\
\cong\ \big(S(f(x;n)),\leqq\big)$. Also since
$f(x;m_1m_2)=f(x;m_1)+f(x;m_2)$ for all $m_1,m_2\in F(n)$, there
exists a one-to-one correspondence between factorizations
$n=m_1m_2\ldots m_r$ of $n$ with partitions
$f(x;m_1)+f(x;m_2)+\cdots +f(x;m_r)=f(x;m_1m_2\ldots m_r)=f(x;n)$
of $f(x;n)$. Thus we have $p^*(n)=p(f(x;n))$.
\end{proof}

\begin{cor}\label{c:pn}
 Let $p$ be a prime number and $n\in\Nat$. Then
$p^*(p^n)=p(n)$.
\end{cor}

\begin{rem}\label{r:rearr}
It is clear that the value of $p^*(n)$ is independent of the
particular primes involved in the prime factorization expression
of $n$, i.e., if $n$ has the prime factorization (\ref{eq:n}) and
$m=q_1^{n_1}q_2^{n_2}\ldots q_k^{n_k}$, where $q_i$ are distinct
primes, then $p^*(m)=p^*(n)$. Thus the polynomial $f(x;n)$ in
Definition \ref{d:fxn} may be different for different arrangement
of primes in the prime factorization expression of $n$. But the
value of $p(f(x;n))$ remains same for each such arrangements. In
particular, $p(2+x)=p^*(2^2.3)=4=p^*(2.3^2)=p(1+2x)$. Indeed, the
distinct partitions of the polynomial $1+2x$ in $\Int_0^+[x]$ are
\begin{eqnarray*}
1+2x&=&(1)+(x)+(x)\\
&=&(1+x)+(x)\\
&=&(1)+(2x)\\
&=&(1+2x)
\end{eqnarray*}
More generally, $p^*(p^2q)=4$ for any pair of distinct primes
$p,q$.
\end{rem}

\section{Generating function and recurrence relations}

Let $n$ be a natural number. We know that the number of
partitions, $p(n)$ of $n$ is given \cite{H} by the following
classical generating function found by Euler:
\begin{eqnarray}
F(x)&=&\frac{1}{(1-x)(1-x^2)(1-x^3)\ldots} \label{eq:genfn1}\\
&=&\prod\limits_{n=1}^{\infty} \frac{1}{1-x^n} \label{eq:genfn2}\\
&=&\prod\limits_{i=1}^{\infty} \Big(1+\sum\limits_{n=1}^{\infty}
x^{ni}\Big) \label{eq:genfn3}\\
 &=&1+\sum\limits_{n=1}^{\infty} p(n) x^n. \label{eq:genfn4}
\end{eqnarray}

In the above equalities, since (\ref{eq:genfn3}) provides all
possible positive integral powers of $x$ less than $n$ (with all
possible multiplicities), product of these terms produce the term
$x^n$ as many times as $n$ can be expressed as a sum of positive
integers which is exactly the number of partitions of $n$, i.e.,
the term $x^n$ occurs $p(n)$ times. So we get the coefficient
$p(n)$ of $x^n$ in (\ref{eq:genfn4}). Similarly, if we wish to
find the number of partitions of the polynomial $f(x)$ in the
polynomial semiring $\ZX$, we have to consider the product of
summations which provides all possible polynomials in $\ZX$ less
than $f(x)$ (with all possible multiplicities) as indices. Thus we
have the following {\em formal} generating function for $p(f(x))$:
\begin{eqnarray}
\mathcal{F}(x)&=&\prod\limits_{{g\,\in\
\ZX}^\star}\frac{1}{1-e^{g(x)}}
 \label{eq:genfnpf1}\\
&=&\prod\limits_{{g\,\in\ \Int_0^+[x]}^\star} \Big(
1+\sum\limits_{n=1}^{\infty}
e^{ng(x)}\Big) \label{eq:genfnpf2}\\
&=&1+\sum\limits_{{f\,\in\ \Int_0^+[x]}^\star} p(f)\ e^{f(x)},
\label{eq:genfnpf3}
\end{eqnarray}
where ${\ZX}^\star=\ZX\smallsetminus \set{0}$.

\begin{rem}
The expressions (\ref{eq:genfnpf1})-(\ref{eq:genfnpf3}) are
merely formal in the sense that for any particular $f(x)\in\ZX$,
the coefficients of $e^{f(x)}$ in either side are same. So the
convergence problem does not arise here. However, if one insists
on it, one may replace $e$ by $e_1=\frac{1}{e}$ in which case
(\ref{eq:genfnpf1})-(\ref{eq:genfnpf3}) are absolutely and
uniformly convergent for all positive integral values of $x$. For
example, consider $\mathcal{F}(1)=\prod\limits_{{g\,\in\
\ZX}^\star}\frac{1}{1-e_1^{g(1)}}$. The product
$\prod\limits_{{g\,\in\ \ZX}^\star} (1-e_1^{g(1)})$ is convergent
if $\sum\limits_{{g\,\in\ \ZX}^\star} e_1^{g(1)}$ is convergent.
Now $\sum\limits_{{g\,\in\ \ZX}^\star}
e_1^{g(1)}=\sum\limits_{n=1}^{\infty} p(n)e_1^n$ which is
absolutely and uniformly convergent for $e_1=\frac{1}{e}$
\cite{H}.
\end{rem}

\vspace{1em} Now (\ref{eq:genfnpf1}) can be written in the form:
\begin{equation} \label{eq:genfnpf4}
\mathcal{F}(x)\ =\ \prod\limits_{\atop{{g\,\in\
\Int_0^+[x]}^\star}{g\textrm{ is
primitive}}}\frac{1}{\prod\limits_{n=1}^{\infty} (1-e^{ng(x)})}
\end{equation}

which is again by (\ref{eq:genfn4}),
\begin{equation} \label{eq:genfnpf5}
\mathcal{F}(x)\ =\ \prod\limits_{\atop{{g\,\in\
\Int_0^+[x]}^\star}{g\textrm{ is primitive}}} \Big(
1+\sum\limits_{n=1}^{\infty} p(n)\ e^{ng(x)}\Big)
\end{equation}

So we have the following generating function for $p(f(x))$:
\begin{equation}
1+\sum\limits_{{f\in\,\Int_0^+[x]}^\star} p(f)\ e^{f(x)}\ =\
\prod\limits_{\atop{{g\in\,\Int_0^+[x]}^\star}{g\textrm{ is
primitive}}} \Big( 1+\sum\limits_{n=1}^{\infty} p(n)\
e^{ng(x)}\Big)
\end{equation}

Using this we describe a method of calculating $p(f(x))$:

\begin{equation}
p(f(x))=\ \textrm{ coefficient of }\ e^{f(x)}\ \textrm{ in }\
\prod\limits_{\atop{0<g\leqslant f,\ g\in\,\Int_0^+[x]}{g\textrm{
is primitive}}} \Big( 1+\sum\limits_{n=1}^{\infty} p(n)\
e^{ng(x)}\Big).
\end{equation}

\begin{exmp}
Let $n=12$. Then $n=2^2.3$ and so the associated polynomial
$f(x;12)=2+x$. Now primitive polynomials less than or equal to
$2+x$ in $\Int_0^+[x]$ are $1,x,1+x$ and $2+x$. So we have
\begin{eqnarray*}
p(f(x;12))&=&p(2+x)\\
&=&\textrm{ coefficient of }\ e^{2+x}\ \textrm{ in }\\
&&\big(1+p(1)e+p(2)e^2\big)\ \big(1+p(1)e^x\big)\
\big(1+p(1)e^{1+x}\big)\
\big(1+p(1)e^{2+x}\big)\\
&=&\textrm{ coefficient of }\ e^{2+x}\ \textrm{ in }\
(1+e+2e^2)(1+e^x)(1+e^{1+x})(1+e^{2+x})\\
&=&\textrm{ coefficient of }\ e^{2+x}\ \textrm{ in }\
(1+e+2e^2)(1+e^x+e^{1+x}+e^{2+x})\\
&=&\textrm{ coefficient of }\ e^{2+x}\ \textrm{ in }\
1+e+2e^2+e^x+2e^{1+x}+4e^{2+x}\\
&=&4.
\end{eqnarray*}
 Thus we get that $p^*(12)=p(f(x;12))=4$.
\end{exmp}

\begin{rem} {\bf (i)}\ \ Note that in each step of the above
calculation, we omit the terms $e^{h(x)}$ whenever $h(x)>
f(x;12)$, as these terms have no further contribution in forming
the term $e^{f(x;12)}$.

\vspace{1em} {\bf (ii)}\ \ By the process of calculating
$p(f(x))$, we are getting all the values of $p(g(x))$ for all
$g(x)\leqslant f(x)$ in $\Int_0^+[x]$. For example, the above
calculation gives us
$$p(1)=1,\ p(2)=2,\ p(x)=1,\ p(1+x)=2.$$
\end{rem}

Next we wish to obtain a recurrence relation for $p(f(x))$. From
(\ref{eq:genfnpf1}) and (\ref{eq:genfnpf3}), we get that
\begin{equation}\label{eq:rec1}
\Big(1+\sum\limits_{{f\,\in\ \Int_0^+[x]}^\star} p(f)\
e^{f(x)}\Big)\ \cdot\ \prod\limits_{{g\,\in\ \ZX}^\star}
\Big(1-e^{g(x)}\Big)\ =\ 1.
\end{equation}
 Now taking
formal derivatives\footnote{The {\em\textbf{formal derivative}} of
a polynomial $h(x)=a_0+a_1x+a_2x^2+\cdots +a_kx^k\in\ZX$ is
defined by $h^\prime(x)=a_1+2a_2x+3a_3x^2+\cdots +ka_kx^{k-1}$.
One can easily extend this definition for formal power series. The
operator derivative is additive and follows Leibnitz' rule. It is
routine to verify that the derivative of
$e^{h(x)}=e^{h(x)}h^\prime(x)$.} on both sides of (\ref{eq:rec1})
we get,

\noindent $\displaystyle{\Big(\sum\limits_{{f\,\in\
\Int_0^+[x]}^\star} p(f)\ e^{f(x)}\cdot f^\prime(x)\Big)\ \cdot\
\prod\limits_{{g\,\in\
\ZX}^\star} \big(1-e^{g(x)}\big)\ +}$\\[0.35em]
\null\hfill $\displaystyle{\Big(1+\sum\limits_{{f\,\in\
\Int_0^+[x]}^\star} p(f)\ e^{f(x)}\Big)\ \cdot
\Big(\sum\limits_{{g\,\in\ \ZX}^\star} \Big\{\big(
-e^{g(x)}g^\prime(x)\big)\prod\limits_{\atop{{g_1\,\in\
\Int_0^+[x]}^\star}{g_1\neq g}} \big(1-e^{g_1(x)}\big)\Big\}\Big)\
=\ 0}$

\noindent which implies
\begin{equation}
\sum\limits_{{f\,\in\ \Int_0^+[x]}^\star} p(f)\ e^{f(x)}\cdot
f^\prime(x)\ =\ \Big(1+\sum\limits_{{f\,\in\ \Int_0^+[x]}^\star}
p(f)\ e^{f(x)}\Big)\cdot \Big( \sum\limits_{{g\,\in\ \ZX}^\star}
\Big\{e^{g(x)} g^\prime(x) \cdot \big(1+\sum\limits_{r=1}^{\infty}
e^{rg(x)}\big)\Big\}\Big).
\end{equation}
Then equating the coefficient of $e^{f(x)}$ of both sides we get
\begin{equation}
p(f)\ f^\prime(x)\ =\ \big(\sum\limits_{r\big| c(f)}
\frac{1}{r}\big)\cdot f^\prime(x)\ +\ \sum\limits_{\atop{0<g< f}{
g\in\,\Int_0^+[x]}} g^\prime(x)\ \Big(\sum\limits_{r=1}^{k_g}
p(f-rg)\Big).
\end{equation}
where $c(f)$ is the content\footnote{i.e., $\gcd$ of all
coefficients  of the polynomial $f(x)$.} of the polynomial $f(x)$
and $k_g=\max\Set{r\in\Nat}{f(x)>rg(x)}$. Considering $p(0)=1$ and
replacing $rg$ by $g$ we finally have
\begin{equation}\label{eq:rec2}
p(f)\ f^\prime(x)\ =\ \sum\limits_{\atop{0<g\leqq f}{
g\in\,\Int_0^+[x]}} \lambda(g)\ g^\prime(x)\ p(f-g).
\end{equation}
where $\lambda(g)=\sum\limits_{i\big| c(g)} \frac{1}{i}$\, . Now
since polynomials on both sides of (\ref{eq:rec2}) are identical,
we may equate coefficients of each power of $x$. So if
$f(x)=n_1+n_2x+n_3x^2+\cdots +n_kx^{k-1}$ and $b_2(g)$ denotes
the coefficient of $x$ in $g(x)$ for each $g\in\ZX$ such that
$0<g(x)\leqq f(x)$, then equating constant terms in
(\ref{eq:rec2}) we have
\begin{equation}\label{eq:rec3}
p(f)\ n_2\ =\ \sum\limits_{\atop{0<g\leqq f}{ g\in\,\Int_0^+[x]}}
\lambda(g)\ b_2(g)\ p(f-g).
\end{equation}
Also by Remark \ref{r:rearr}, one may rearrange coefficients of
$f(x)$ which does not change the value of $p(f)$. Thus for each
$i=1,2,\ldots ,k$, we may rearrange coefficients of $f(x)$ in
such a way that $n_i$ will be the coefficient of $x$. Then we have
\begin{equation}\label{eq:rec4}
p(f)\ n_i\ =\ \sum\limits_{\atop{0<g\leqq f}{ g\in\,\Int_0^+[x]}}
\lambda(g)\ b_i(g)\ p(f-g),
\end{equation}
where $b_i(g)$ denotes the coefficient of $x^{i-1}$ in $g(x)$ for
each $g\in\ZX$ such that $0<g(x)\leqq f(x)$. Therefore suitably
multiplying the above relations by powers of $x$ and adding we get
\begin{equation}\label{eq:rec5}
p(f)\ f(x)\ =\ \sum\limits_{\atop{0<g\leqq f}{ g\in\,\Int_0^+[x]}}
\lambda(g)\ g(x)\ p(f-g),
\end{equation}

\begin{rem}
{\bf (i)}\ We first note that (\ref{eq:rec5}) is a nice
generalization of Euler's recurrence relation for $p(n),\
(n\in\Nat)$ which is given by
\begin{equation}\label{eq:euler}
n\, p(n)=\sum\limits_{j=1}^{n} \sigma(j)\, p(n-j),
\end{equation}
 where
$\sigma(j)$ denotes the sum of all divisors of $j$ and $p(0)$ is
assumed to be $1$. Now (\ref{eq:euler}) is immediately obtained
from (\ref{eq:rec5}) by putting $f(x)=n$ (the constant
polynomial), as we already have, by Corollary \ref{c:pn},
$p(n)=p^*(p^n)$. Note that $\lambda(j)j=\sigma(j)$ for all
$j=1,2,\ldots ,n$.

\vspace{1em} {\bf (ii)}\ The recurrence relation (\ref{eq:rec5})
also generalizes the one obtained by Harris and Subbarao
\cite{HS}. In Remark 2 of \cite{HS} (pg.477), they described an
equivalent form of their recurrence relation as follows:

{\small Consider the vector $\vec{\A}(n)=(\A_1,\A_2,\ldots \A_k)$
for the natural number $n=q_1^{\A_1}q_2^{\A_2}\ldots q_k^{\A_k}$.
Then
\begin{equation}\label{eq:rec6}
p^*(\vec{\A})\ \|\vec{\A}\|\ =\ \sum\limits_{\vec{0}\le
\vec{\B}\le\vec{\A}} p^*(\vec{\A}-\vec{\B})
\lambda(\vec{\B})\|\vec{\B}\|,
\end{equation}
where $p^*(\vec{\A})=p^*(n),\ \|\vec{\A}\|=\prod_{j=1}^{k} \A_j,\
\lambda(\vec{\A})=\sum_{i:i\big| \A_j \textrm{ for } 1\le j\le
k\}} 1/i$ and $\vec{\B}<\vec{\A}$ means that $\B_j\le \A_j$ for
$1\le j\le k$.}

We first note that the limit under sum in (\ref{eq:rec6}) should
be $\vec{0}< \vec{\B}\le\vec{\A}$ as $\lambda(\vec{0})$ is not
defined. Secondly, $\|\vec{\A}\|$ will be $\sum_{j=1}^{k} \A_j$,
as one may verify $p^*(18)=5$ by (\ref{eq:rec6}), which is wrong.
Finally, while defining the ordering for vectors one has to use
$\le$ instead of $<$. However keeping aside these printing
mistakes, the correct version of (\ref{eq:rec6}) is given by
\begin{equation}\label{eq:rec7}
p^*(\vec{\A})\ \|\vec{\A}\|\ =\ \sum\limits_{\vec{0}<
\vec{\B}\le\vec{\A}} p^*(\vec{\A}-\vec{\B})
\lambda(\vec{\B})\|\vec{\B}\|,
\end{equation}
where $\|\vec{\A}\|=\sum\limits_{j=1}^{k} \A_j$. Now
(\ref{eq:rec7}) immediately follows from (\ref{eq:rec5}) by
putting $x=1$ (or, adding equations (\ref{eq:rec4}) for
$i=1,2,\ldots,k$).
\end{rem}

\vspace{1em} Now for any $n\in\Nat$ with prime factorization
(\ref{eq:n}), we define $c(n)=\gcd\set{n_1,n_2,\ldots n_k}$ and
$\lambda(n)=\sum\limits_{i\big| c(n)} \frac{1}{i}$. Then equating
constant terms in (\ref{eq:rec5}) we get
\begin{equation}\label{eq:rec8}
p(f)\ n_1\ =\ \sum\limits_{\atop{0<g\leqq f}{ g\in\,\Int_0^+[x]}}
\lambda(g)\ b_1(g)\ p(f-g),
\end{equation}
where $b_1(g)$ denotes the constant term of the polynomial $g(x)$
for each $g\in\ZX$ such that $0<g(x)\leqq f(x)$. This implies
\begin{equation}\label{eq:rec9}
p^*(n)\ n_1\ =\ \sum\limits_{d\big| n,\ p_1\big| d}
r_1\lambda(d)\ p^*\big(\frac{n}{d}\big),
\end{equation}
where the summation runs over $d=p_1^{r_1}p_2^{r_2}\ldots
p_k^{r_k}$ for $0<r_1\leqslant n_1$ and $0\leqslant r_j\leqslant
n_j,\ j=2,3,\ldots ,k$. Thus we have
\begin{equation}\label{eq:rec10}
p^*(n)\ n_1\ =\ \sum\limits_{i=1}^{n_1} i\cdot
\Big(\sum\limits_{d\big| \frac{n}{p_1^{n_1}}} \lambda(p_1^i d)\
p^*\big(\frac{n}{p_1^id}\big)\Big).
\end{equation}

\begin{rem}
We note that (\ref{eq:rec10}) improves the recurrence relation
obtained by Harris and Subbarao \cite{HS} as it contains less
terms. In fact, the number of terms in (\ref{eq:rec10}) is
$(n_2+1)(n_3+1)\ldots (n_k+1)$ less than that of (\ref{eq:rec7})
for $n=p_1^{n_1}p_2^{n_2}\ldots p_k^{n_k}$. Thus, in view of
Remark \ref{r:rearr}, it is advisable to arrange prime factors of
$n$ in such a way that $n_1$ should be minimum among all $n_i$'s
for quicker computation.
\end{rem}

We summarize the above results in the following:

\begin{thm}\label{t:main}
Let $n=p_1^{n_1}p_2^{n_2}\ldots p_k^{n_k}$ be the prime
factorization of a natural number $n$, where $p_i$ are distinct
primes, $n_i\in\Nat$ for all $i=1,2,\ldots ,k,\ (k\in\Nat)$, with
$n_1=\min\Set{n_i}{i=1,2,\ldots ,k}$. Then
\begin{equation}\label{eq:rec11}
p(f)\ f(x;n)\ =\ \sum\limits_{\atop{0<g\leqq f}{
g\in\,\Int_0^+[x]}} \lambda(g)\ g(x)\ p(f-g),
\end{equation}
where $f(x;n)=n_1+n_2x+n_3x^2+\cdots +n_kx^{k-1},\
\lambda(g)=\sum\limits_{i\big| c(g)} \frac{1}{i},\ c(g)$ being
the content of the polynomial $g(x)$ for each $g\in\ZX$ such that
$0<g(x)\leqq f(x)$. In particular,
\begin{equation}\nonumber
p^*(n)\ n_1\ =\ \sum\limits_{i=1}^{n_1} i\cdot
\Big(\sum\limits_{d\big| \frac{n}{p_1^{n_1}}} \lambda(p_1^i d)\
p^*\big(\frac{n}{p_1^id}\big)\Big),
\end{equation}
where $\lambda(m)=\sum\limits_{i\big| c(m)} \frac{1}{i},\
c(m)=\gcd\set{i,r_2,\ldots r_k}$ for $m=p_1^{i}p_2^{r_2}\ldots
p_k^{r_k},\ 0<i\leqslant n_1$ and $0\leqslant r_j\leqslant n_j,\
j=2,3,\ldots ,k$.
\end{thm}

A simple formula is obtained the case of $n_1=1$. In that case,
$c(m)=1$ for all $m\big| n$ with $p_1\big| m$ and So
$\lambda(m)=1$ for all such $m$. Thus we have

\begin{cor}\label{c:main}
Let $n$ be a natural number and $p$ be a prime number such that
$p\not\big|\,\, n$. Then
\begin{equation}\label{eq:rec12}
p^*(np)=\sum\limits_{d\big| n} p^*(d).
\end{equation}
In particular, for any two distinct prime numbers $p,q$ and for
any natural number $n$,
\begin{equation}\label{eq:rec13}
p^*(pq^n)=\sum\limits_{i=0}^{n} p(i),
\end{equation}
\end{cor}

\begin{exmp}
Let $n=72$. Then $n=2^3.3^2=3^2.2^3$. Now the divisors of $2^3$
are $1,2,4,8$. So we have
\begin{eqnarray*}
&&p^*(72)\cdot 2\\
&=&\sum\limits_{i=1}^{2} i\cdot
\Big(\sum\limits_{d\big| 8} \lambda(3^i d)\
p^*\big(\frac{72}{3^id}\big)\Big),\\
&=& \big\{p^*(24)+ p^*(12)+p^*(6)+ p^*(3)\big\}+ \\
&&\null\hspace{1.5in} 2\cdot \big\{\lambda(3^2)\ p^*(8)+
\lambda(3^2\cdot 2)\ p^*(4)+ \lambda(3^2\cdot 2^2)\ p^*(2)+
\lambda(3^2\cdot 2^3)\ p^*(1)\big\}.
\end{eqnarray*}

Now $p^*(2)=p^*(3)=1,\ p^*(4)=p^*(2^2)=p(2)=2,\
p^*(8)=p^*(2^3)=p(3)=3$ by Corollary \ref{c:pn}. Again
$p^*(6)=p^*(3\cdot 2)=1+p(1)=2,\ p^*(12)=p^*(3\cdot
2^2)=2+p(2)=4,\ p^*(24)=p^*(3\cdot 2^3)=4+p(3)=7$ by
(\ref{eq:rec13}). Thus we have
$$p^*(72)\cdot 2=\big\{7+4+2+1\}+2\cdot \big\{(1+\frac{1}{2})\cdot
3+1\cdot 2+(1+\frac{1}{2})\cdot 1+1\cdot 1\big\}=14+18=32.$$
Hence $p^*(72)=16$. Indeed $16$ factorizations of $72$ are
\begin{eqnarray*}
72&=&2.2.2.3.3\ =\ 2.2.2.9\ =\ 2.2.3.6\ =\ 2.2.18\
=\ 2.3.3.4\ =\ 2.3.12\ =\ 2.4.9\\
&=&2.6.6\ =\ 2.36\ =\ 3.3.8\ =\ 3.4.6\ =\ 3.24\ =\ 4.18\ =\ 6.12\
=\ 8.9\ =\ 72.
\end{eqnarray*}
\end{exmp}

% \vfill
\vspace{0em}
%GATHER{Xbib.bib}   % For Gather Purpose Only
%GATHER{Paper.bbl}  % For Gather Purpose Only
\renewcommand{\baselinestretch}{1}
\bibliographystyle{amsplain}
{\small \bibliography{fact}} %give some bibtex file name here

\providecommand{\bysame}{\leavevmode\hbox to3em{\hrulefill}\thinspace}
\providecommand{\MR}{\relax\ifhmode\unskip\space\fi MR }
\providecommand{\MRhref}[2]{%
  \href{http://www.ams.org/mathscinet-getitem?mr=#1}{#2}
}
\providecommand{\href}[2]{#2}
\begin{thebibliography}{1}

\bibitem{D}
D.~Beckwith, \emph{Problem 10669}, Amer. Math. Monthly \textbf{105} (1998),
  559.

\bibitem{CEP}
E.~R. Canfield, Paul Erd\"{o}s, and Carl Pomerance, \emph{On a problem of
  oppenheim concerning ``factorisatio numerorum''}, J. Number Theory
  \textbf{17} (1983), 1--28.

\bibitem{GN}
R.~K. Guy and R.~J. Nowakowski, \emph{Monthly unsolved problems, 1969-1995},
  Amer. Math. Monthly \textbf{102} (1995), 921--926.

\bibitem{H}
G.~H. Hardy and E.~M. Wright, \emph{An introduction to the theory of numbers},
  The English Language Book Society and Oxford University Press, 1981.

\bibitem{HS}
V.~C. Harris and M.~V. Subbarao, \emph{On product partitions of inetegers},
  Canad. Math. Bull. \textbf{34} (1991), no.~4, 474--479.

\bibitem{HuS}
John~F. Hughes and J.~O. Shallit, \emph{On the number of multiplicative
  partitions}, Amer. Math. Monthly \textbf{90} (1983), 468--471.

\bibitem{MD}
L.~E. Mattics and F.~W. Dodd, \emph{A bound for the number of multiplicative
  partitions}, Amer. Math. Monthly \textbf{93} (1986), 125--126.

\bibitem{OP}
A.~Oppenheim, \emph{On an arithmetic function}, J. London Math. Soc. \textbf{1}
  (1926), 205--211.

\bibitem{ST}
G.~Szekeres and P.~Tur\'{a}n, \emph{\"{U}ber das zweite hauptproblem der
  ``factorisatio numerorum''}, Acta Litt. Sci. Szeged \textbf{6} (1933),
  143--154.

\end{thebibliography}

\end{document}